\documentclass[a4paper,UKenglish,cleveref, autoref, thm-restate]{lipics-v2021}
\pdfoutput=1

\usepackage{amsmath,amssymb,amsthm,mathtools,comment}
\usepackage{enumerate}
\usepackage{float}
\usepackage{graphicx}
\usepackage{todonotes}
\usepackage{complexity}
\usepackage{tikz}
\usetikzlibrary{automata,positioning, decorations.pathreplacing, arrows, decorations.markings}

\def\vecsign{\mathchar"017E}
\def\dvecsign{\smash{\stackon[-1.95pt]{\vecsign}{\rotatebox{180}{$\vecsign$}}}}
\def\dvec#1{\def\useanchorwidth{T}\stackon[-4.2pt]{#1}{\,\dvecsign}}
\usepackage{stackengine}
\stackMath

\theoremstyle{definition}

\newtheorem*{lemma*}{Lemma}
\newtheorem*{theorem*}{Theorem}
\newtheorem{invariant}{Invariant}

\newcommand{\calA}{\mathcal{A}}

\newcommand{\calC}{\mathcal{C}}

\newcommand{\Fbb}{\mathbb{F}}
\newcommand{\N}{\mathbb{N}}
\newcommand{\adom}{\textnormal{adom}}

\title{Dynamic Meta-theorems for Distance and Matching} 

\author{Samir Datta}{Chennai Mathematical Institute, Chennai, India}{sdatta@cmi.ac.in}{}{Partially funded by a grant from Infosys foundation and
SERB-MATRICS grant MTR/2017/000480}
\author{Chetan Gupta}{Aalto University, Finland}{chgpt.09@gmail.com}{}{Supported by Academy of Finland, Grant 321901}
\author{Rahul Jain}{Fernuniversit\"at in Hagen, Germany}{rahul.jain@fernuni-hagen.de}{}{}
\author{Anish Mukherjee}{Institute of Informatics, University of Warsaw, Poland}{anish@mimuw.edu.pl}{}{Supported by the ERC CoG grant TUgbOAT no 772346}
\author{Vimal Raj Sharma}{Indian Institute of Technology, Kanpur, India}{vimalraj@cse.iitk.ac.in}{}{Ministry of Electronics and IT, India, VVS PhD program}
\author{Raghunath Tewari}{Indian Institute of Technology, Kanpur, India}{rtewari@cse.iitk.ac.in}{}{Young Faculty Research Fellowship, Ministry of Electronics and IT, India}

\authorrunning{S. Datta, C. Gupta, R. Jain, A. Mukherjee, V.R. Sharma, R. Tewari} 

\Copyright{Samir Datta, Chetan Gupta, Rahul Jain, Anish Mukherjee, Vimal Raj Sharma and Raghunath Tewari} 



\acknowledgements{SD and AM would like to thank Nils Vortmeier and Thomas Zeume for their comments on Section~\ref{sec:rank}.}

\newclass{\ReachUL}{ReachUL}
\newclass{\coUL}{coUL}

\newlang{\NZCL}{NonZeroCircL}
\newlang{\NZCNC}{NonZeroCircNC}

\newclass{\Log}{L}
\newclass{\ACz}{AC^0}
\newclass{\TCz}{TC^0}
\newclass{\ACo}{AC^1}
\newclass{\ACzt}{AC^0[\oplus]}
\newclass{\FOar}{FO(\le,+,\times)}
\newclass{\FOpar}{FO[\oplus](\le,+,\times)}
\newclass{\DynACz}{DynAC^0}
\newclass{\DynTCz}{DynTC^0}
\newclass{\DynACzt}{DynAC^0[\oplus]}
\renewclass{\DynFO}{DynFO}
\newclass{\DynFOar}{DynFO(\le,+,\times)}
\newclass{\DynFOp}{DynFO[\oplus]}
\newclass{\DynFOpar}{DynFO[\oplus](\le,+,\times)}
\newclass{\DynFOmar}{DynFO[MAJ](\le,+,\times)}

\newlang{\PM}{PM}
\newlang{\BPM}{BPM}
\newlang{\PMD}{PMDecision}
\newlang{\BPMD}{BPMDecision}
\newlang{\PMS}{PMSearch}
\newlang{\BPMS}{BPMSearch}
\newlang{\BMWPMS}{MinWtBPMSearch}
\newlang{\MCM}{MM}
\newlang{\BMCM}{BMM}
\newlang{\BMCMD}{BMMDecision}
\newlang{\BMCMS}{BMMSearch}
\newlang{\MCMSz}{MMSize}
\newlang{\BMCMSz}{BMMSize}
\newlang{\MWMCM}{MinWtMM}
\newlang{\BMWMCM}{MinWtBMM}
\newlang{\BMWMCMS}{MinWtBMMSearch}
\newlang{\Reach}{Reach}
\newlang{\Dist}{Distance}
\newlang{\Rank}{Rank}

\nolinenumbers 

\hideLIPIcs  

\EventEditors{John Q. Open and Joan R. Access}
\EventNoEds{2}
\EventLongTitle{42nd Conference on Very Important Topics (CVIT 2016)}
\EventShortTitle{CVIT 2016}
\EventAcronym{CVIT}
\EventYear{2016}
\EventDate{December 24--27, 2016}
\EventLocation{Little Whinging, United Kingdom}
\EventLogo{}
\SeriesVolume{42}
\ArticleNo{23}

\begin{document}
\maketitle

\begin{abstract}
Reachability, distance, and matching are some of the most fundamental graph problems that have been of particular interest in dynamic complexity theory in recent years \cite{DKMSZ18, DMVZ18, DKMTVZ20}. Reachability can be maintained with first-order update formulas, or equivalently in $\DynFO$ in general graphs with $n$ nodes~\cite{DKMSZ18}, even under $O(\frac{\log{n}}{\log{\log{n}}})$ changes per step \cite{DMVZ18}. In the context of how large the number of changes can be handled, it has recently been shown \cite{DKMTVZ20} that under a polylogarithmic number of changes, reachability is in $\DynFOpar$ in planar, bounded treewidth, and related graph classes -- in fact in any graph where small \emph{non-zero circulation weights} can be computed in \NC. 

We continue this line of investigation and extend the meta-theorem for reachability to distance and bipartite maximum matching with the same bounds. These are amongst the most general classes of graphs known where we can maintain these problems deterministically without using a majority quantifier and even maintain witnesses. For the bipartite matching result, modifying the approach from \cite{FGT}, we convert the static non-zero circulation weights to dynamic matching-isolating weights. 

While reachability is in $\DynFOar$ under $O(\frac{\log{n}}{\log{\log{n}}})$ changes, no such bound is known for either distance or matching in any non-trivial class of graphs under non-constant changes. We show that, in the same classes of graphs as before, bipartite maximum matching is in $\DynFOar$ under $O(\frac{\log{n}}{\log{\log{n}}})$ changes per step. En route to showing this we prove that the rank of a matrix can be maintained in $\DynFOar$, also under $O(\frac{\log{n}}{\log{\log{n}}})$ entry changes, improving upon the previous $O(1)$ bound \cite{DKMSZ18}. This implies similar extension for the non-uniform $\DynFO$ bound for maximum matching in general graphs and an alternate algorithm for maintaining reachability under $O(\frac{\log{n}}{\log{\log{n}}})$ changes \cite{DMVZ18}.
\end{abstract}
\section{Introduction}
In real-life scenarios,  many situations involve an evolving input where parts of the data change frequently. Recomputing everything from scratch for these large datasets after every change is not an efficient option, and the goal is to dynamically maintain some auxiliary data structure to help us recompute the results quickly.  

The dynamic complexity framework of Patnaik and Immerman~\cite{PatnaikI} is one such approach that has its roots in descriptive complexity \cite{Immerman} and is closely related to the setting of Dong,  Su,  and Topor~\cite{DongST95}.  For example,  by maintaining some auxiliary relations,  the reachability relation can be updated after every single edge modification using first-order logic formulas~\cite{DKMSZ18} i.e., the reachability query is contained in the dynamic complexity class $\DynFO$~\cite{PatnaikI}.  The motivation to use first-order logic as the update method has connections to other areas.  From the circuit complexity perspective,  this implies that such queries are highly parallelizable, i.e., can be updated by polynomial-size circuits in constant-time due to the correspondence between $\FO$ and uniform $\ACz$ circuits \cite{BIS}. From the perspective of database theory, such a program can be translated into equivalent \textsc{SQL} queries.

The area has seen renewed interest in proving further upper bounds results, partly after the resolution of the long-standing conjecture~\cite{PatnaikI} that reachability is in $\DynFO$ under single edge modifications~\cite{DKMSZ18}. A natural direction to extend this result is to see which other fundamental graph problems also admit such efficient dynamic programs. The closely related problems of maintaining distance and matching are two such examples, though a $\DynFO$ bound for these problems in general graphs has been elusive so far. The best known bound for distance is $\DynTCz$~\cite{Hesse03} and non-uniform $\DynACzt$ \cite{DHK}.  Here, the updates are computed in $\FO$ formulas with majority quantifiers (uniform $\TCz$ circuits) and non-uniform $\FO$ formulas with parity quantifiers ($\ACzt$ circuits), respectively. For matching, we have a non-uniform $\DynFO$ bound for maintaining the size of the maximum matching~\cite{DKMSZ18}. The only non-trivial class of graphs where both these problems are in $\DynFO$ is bounded treewidth graphs~\cite{DMSVZ19}. 

At the same time progress has been made to understand how large a modification to the input can be handled by similar dynamic programs. It is of particular interest since, in applications, changes to a graph often come as a bulk set of edges. It was shown that reachability can be maintained in $\DynFOar$ under changes of size $O(\log n/ \log \log n)$~\cite{DMVZ18} in graphs with $n$ nodes.  Here, the class $\DynFOar$ extends $\DynFO$ by access to built-in arithmetic, which is more natural for bulk changes. To handle larger changes, it is known that even for reachability, changes of size larger than polylogarithmic cannot be handled in $\DynFO$~\cite{DKMTVZ20}. And for changes of polylogarithmic size, the previous techniques seem to require extending $\DynFO$ by majority quantifiers~\cite{DMVZ18}. Under bulk changes of polylog($n$) size, we can even maintain distance and the size of a maximum matching in the uniform and the non-uniform version of $\DynTCz$, respectively~\cite{DMVZ18, Mukherjee}. 

Making further progress in this direction, recently it has been shown in \cite{DKMTVZ20} that reachability is in $\DynFOpar$ (i.e., update formulas may use parity quantifiers) under polylog($n$) changes in the class of graphs where polynomially bounded \emph{non-zero circulation weights} can be computed statically in the parallel complexity class \NC. A weight function for the edges of a graph has non-zero circulation if the (alternate) sum of the weight of every directed cycle is non-zero (see Section~\ref{sec:prelims} for more details). Planar \cite{TV12}, bounded genus \cite{DKTV11}, bounded treewidth graphs \cite{DKMTVZ20} are some of the well-studied graph classes for which such non-zero circulation weights can be computed in \NC. 

In this work, we first extend this result to prove similar meta-theorems for maintaining distance (including a shortest path witness) and the \emph{search} version of minimum weight bipartite maximum matching ($\BMWMCMS$) in the same class of graphs. 
\begin{theorem}\label{thm:dynFO2}
$\Dist$ and $\BMWMCMS$ are in $\DynFOpar$ under polylog$(n)$ edge changes on classes of graphs where non-zero circulation weights can be computed in \NC.
\end{theorem}
Note that these are the only classes of graphs known where we can maintain both these problems \emph{deterministically} without using a majority quantifier and even maintain a witness to the solution. 

While reachability can be maintained in $\DynFOar$ under bulk changes even in general graphs, no such bound is known for either distance or matching in any non-trivial class of graphs under a non-constant number of changes. Extending the dynamic rank algorithm of \cite{DKMSZ18} to handle bulk changes (Theorem~\ref{thm:rank}), we show a similar meta theorem for maintaining the size of a bipartite maximum matching ($\BMCMSz$) in $\DynFOar$ under slightly sublogarithmic bulk changes, in the same class of graphs. 
\begin{theorem}\label{thm:matchFO}
$\BMCMSz$ is in $\DynFOar$ under $O(\frac{\log{n}}{\log{\log{n}}})$ edge changes on classes of graphs for which non-zero circulation weights can be computed in \NC.
\end{theorem}
Previously, no $\DynFOar$ bound was known even in planar graphs under single edge changes. Theorem~\ref{thm:matchFO} also extends the non-uniform $\DynFO$ bound for maintaining the size of maximum matching in \emph{general} graphs from single edge changes \cite{DKMSZ18} to $O(\frac{\log{n}}{\log{\log{n}}})$ many. 

Since reachability reduces to maintaining the rank of a matrix via bounded-expansion first-order (bfo) reductions \cite{DKMSZ18}, Theorem~\ref{thm:rank}  also gives an alternative algorithm for maintaining reachability in $\DynFOar$ under $O(\frac{\log{n}}{\log{\log{n}}})$ changes. This is interesting in its own right as it generalizes the \emph{rank-method} for maintaining reachability \cite{DKMSZ18} even under bulk changes without going via the Sherman-Morrison-Woodbury identity~\cite{DMVZ18}.
%
\begin{table}[h]
\centering
\begin{tabular}{|c|c|c|c|}
\cline{1-4}
\multicolumn{1}{|c|}{Problem} & \multicolumn{3}{|c|}{\#changes}\\
\cline{2-4}
 & $O(1)$ & $O(\frac{\log{n}}{\log{\log{n}}})$ & $\log^{O(1)}{n}$\\
\cline{1-4}
 $\Reach$	& $\DynFO$\ \cite{DKMSZ18} & $\DynFO$\ \cite{DMVZ18} & 
 $\DynFOp$\ \cite{DKMTVZ20} \\
  $\Dist$ 	& \framebox{$\DynFOp$} & \framebox{$\DynFOp$} & \framebox{$\DynFOp$} \\
 $\BMCMSz$	& \framebox{$\DynFO$} & \framebox{$\DynFO$} & \framebox{$\DynFOp$} \\
 $\BMCMS$	&  \framebox{$\DynFOp$} & \framebox{$\DynFOp$} & \framebox{$\DynFOp$} \\
\cline{1-4}
\end{tabular}
\caption{Previously known and \framebox{new} results in graphs with non-zero circulation weights in \NC.}
\end{table}
\vspace*{-8mm}
\subparagraph*{Main Technical Contributions}
There are two major technical contributions of this work:
(1) \emph{Converting the statically computed non-zero circulation weights for
bipartite matchings to dynamically isolating weights for bipartite matchings}. Our main approach (described in detail in Section~\ref{sec:iso}) is to assign polynomially bounded \emph{isolating} weights to the edges of the evolving graph so that the minimum weight solution under these weights is \emph{unique}. While static non-zero circulation weights guarantee this under deletions, for insertions, the dynamization is based on the seminal work of \cite{FGT}. They construct isolating weights for perfect matching for arbitrary bipartite graphs, but which are quasipolynomially large in the size of the graph. By assigning such weights only to the changed part of the graph and carefully combining with the previously assigned weights, we make sure the edge weights remain small as well as isolating throughout, using the Muddling Lemma (see Section~\ref{sec:prelims}). Our construction parallels that of \cite{DKMTVZ20} where dynamic isolating weights for reachability in non-zero circulation graphs were constructed based on the static construction from \cite{KT16}. In addition to extending the reachability result (Theorem~\ref{thm:dynFO2}) this also enables us to prove a $\DynFOar$ bound (Theorem~\ref{thm:matchFO}) for bipartite maximum matching (previously, a rather straightforward application of non-zero circulation weights in planar graphs could only achieve a $\DynFOp$ bound under single edge changes \cite{Mukherjee}). 
2) \emph{Maintaining rank of a matrix under sublogarithmically many changes}. This involves non-trivially extending the technique from \cite{DKMSZ18}, which maintains rank under single entry changes, and combining it with \cite{DMVZ18} which shows how to compute the determinant of a small matrix of dimension $O(\frac{\log{n}}{\log{\log{n}}})$ in $\FOar$. 

\vspace*{-2mm}
\subparagraph*{Organization}
After some preliminaries in Section~\ref{sec:prelims} and Appendix~\ref{sec:prelimsApp}, 
in Section~\ref{sec:iso}, we discuss the connection between dynamic isolation and static non-zero circulation and show its applications for distance and matching in Appendix~\ref{sec:distance} and Section~\ref{sec:ac0p}, respectively.  
In Section~\ref{sec:matchFO}, we describe the $\DynFOar$ algorithm for maximum matching, which is built on
the rank algorithm under bulk changes from Section~\ref{sec:rank}. Finally, we conclude with Section~\ref{sec:concl}.

\section{Preliminaries and Notations}\label{sec:prelims}

%
%
\subparagraph*{Dynamic Complexity}
The goal of a dynamic program is to answer a given query on an \emph{input structure} subjected to insertion or deletion of tuples. The program may use an \emph{auxiliary data structure} over the same domain. Initially, both input and auxiliary structures are empty; and the domain is fixed during each run of the program. 

For a (relational) structure $\mathcal{I}$ over domain $D$ and schema $\sigma$, a change $\Delta \mathcal{I}$ consists of sets $R^{+}$ and $R^{-}$ of tuples for each relation symbol $R \in \sigma$. The result $\mathcal{I} + \Delta \mathcal{I}$ is the input structure where $R^{\mathcal{I}}$ is changed to $(R^{\mathcal{I}} \cup R^{+}) \setminus R^{-}$. The set of affected elements is the (active) domain of tuples in $\Delta \mathcal{I}$. A dynamic program $\mathcal{P}$ is a set of first-order formulas specifying how auxiliary relations are updated after a change. For a state $\mathcal{S} = (\mathcal{I}, \mathcal{A})$ with input structure $\mathcal{I}$ and auxiliary structure $\mathcal{A}$ we denote the state of the program after applying a change sequence $\alpha$ and updating the auxiliary relations accordingly by $\mathcal{P}_\alpha(\mathcal{S})$. 

The dynamic program \emph{maintains} a $q$-ary query $Q$ under changes that affect $k$ elements if it has a $q$-ary auxiliary relation $\textsc{ans}$ that at each point stores the result of $Q$ applied to the current input structure. More precisely, for each non-empty sequence $\alpha$ of changes affecting $k$ elements, the relation $\textsc{ans}$ in $\mathcal{P}_\alpha(\mathcal{S}_\emptyset)$ and $Q(\alpha(\mathcal{I}_\emptyset))$ coincide, where $\mathcal{I}_\emptyset$ is an empty input structure, $\mathcal{S}_\emptyset = (\mathcal{I}_\emptyset, \mathcal{A}_\emptyset)$ where $\mathcal{A}_\emptyset$ denotes the empty auxiliary structure over the domain of $\mathcal{I}_\emptyset$, and $\alpha(\mathcal{I}_\emptyset)$ is the input structure after applying $\alpha$. 

If a dynamic program maintains a query, we say that the query is in \DynFO. Similar to \DynFO\, one can define the class of queries \DynFOar\ that allows for auxiliary relations initialized as a linear order, and the corresponding addition and multiplication relations. One can further extend this class by allowing parity quantifiers to yield the class $\DynFOpar$. As we focus on changes of non-constant size, we include arithmetic in our setting. See \cite{DMVZ18, DKMTVZ20} for more details. 

The \emph{Muddling Lemma}~\cite{DKMTVZ20} states that to maintain many natural queries, it is enough to maintain the query for a bounded number of steps, that we crucially use in this paper. In the following, we first recall the necessary notions before stating the lemma.

A query $Q$ is \emph{almost domain-independent} if there is a $c \in \N$ such that $Q(\calA)[(\adom(\calA) \cup B)] = Q(\calA[(\adom(\calA) \cup B)])$ for all structures $\calA$ and sets $B \subseteq A \setminus \adom(\calA)$ with $|B| \geq c$. Here, $\adom(\calA)$ denotes the \emph{active domain}, the set of elements that are used in some tuple of $\calA$. A query $Q$ is \emph{$(\calC,f)$-maintainable}, for some complexity class $\calC$ and some function~\mbox{$f:\N\to \mathbb{R}$}, if there is a dynamic program $\mathcal{P}$ and a $\calC$-algorithm $\mathbb A$ such that for each input structure $\mathcal{I}$ over a domain of size $n$, each linear order $\leq$ on the domain, and each change sequence $\alpha$ of length $|\alpha| \leq f(n)$, the relation $Q$ in $\mathcal{P}_\alpha(\mathcal{S})$ and $Q(\alpha(\mathcal{I}))$ coincide, where $\mathcal{S} = (\mathcal{I}, \mathbb A(\mathcal{I},\leq))$. $\AC^i$ is the class of problems that can be solved using polynomial-size circuit of $O(\log^i n)$ depth and $\NC\ = \cup_i \AC^i$. 
\begin{lemma}[\cite{DKMTVZ20}]\label{lem:muddle}
Let $Q$ be an almost domain independent query, and let $c \in \N$ be arbitrary. If the query $Q$ is $(\AC^d, 1)$-maintainable under changes of size $\log^{c+d} n$ for some $d \in \N$, then $Q$ is in $\DynFOar$ under changes of size $\log^c n$.
\end{lemma}  
The parity and majority functions of $n$ bits $a_1, \ldots, a_n$ are true if $\sum_{i=1}^n a_i = 1$ (mod $2$) and $\sum_{i=1}^n a_i \ge n/2$, respectively. We refer the readers to Section~\ref{sec:dynprelims} and \cite{DKMSZ18, DMVZ18, DKMTVZ20} for more discussions on the basics of the dynamic complexity framework. 

\subparagraph*{Weight function and Circulation:} Let $G(V, E)$ be an undirected graph with the vertex set $V$ and edge set $E$. By $G(V, \dvec{E})$ denote the  corresponding graph where each of its edges is replaced by two directed edges, pointing in opposite directions.

A set system $\mathcal{M}$ on a universe $U$ is a family of subsets of $U$ i.e.
$\mathcal{M} \subseteq 2^U$. Examples include the family of $s,t$-shortest 
paths and perfect matchings in a graph.
A weight function $w: U \rightarrow \mathbb{Z}_{\geq 0}$ (
the set of non-negative integers)
induces a weight of $w(M) = \sum_{e \in M}{w(e)}$ on an element $M \in \mathcal{M}$. Such a weight function is said to be \emph{isolating} for $\mathcal{M}$ if 
an element $M_0 \in \mathcal{M}$ with the least weight is unique. The notion
of isolation can be extended to a collection of families of graphs such as
the collection of all families of $s,t$-shortest paths for all $s,t \in V$. The weight function $w(e) = 2^e$ is a trivial isolating function -- the crucial point is to give weight function polynomially bounded in the size of the universe.
For an arbitrary set system $\mathcal{M}$ 
Mulmuley et al.~\cite{MulmuleyVV87} gave a randomized construction of such a weight function. 

A weight function $w : \dvec{E} \rightarrow \mathbb{Z}$ is called \emph{skew-symmetric} if for all $e\in \dvec{E}$, $w(e) = -w(e^r)$ (where $e^r$ represent the edge with its direction reversed). 
The \emph{circulation} of a directed cycle under a skew symmetric weight
function is the sum of weights of the directed edges in the cycle. The 
skew-symmetric weight function $w$ induces a \emph{non-zero circulation}
on the graph if every directed cycle in the graph gets a non-zero circulation
under $w$.

We know from \cite{BTV09} that if $w$ assigns non-zero circulation to every cycle that consists of edges of $\dvec{E}$, then it isolates a directed path between each pair of vertices in $G(V,\dvec{E})$. Also, if $G$ is a bipartite graph, then the weight function $w$ can be used to construct a weight function $w' : E \rightarrow \mathbb{Z}$ that isolates a perfect matching 
in $G$ \cite{TV12}. For planar \cite{TV12}, bounded genus \cite{DKTV11}, and bounded treewidth graphs \cite{DKMTVZ20} non-zero circulation weights can be computed \emph{deterministically} in Logspace, which is a subclass of \NC. 

A convention is to represent by $\left<w_1,\ldots,w_k\right>$ the weight function that on edge $e$ takes weight $\sum_{i=1}^k{w_i(e)B^{k-i}}$, where $w_1,\ldots,w_k$ are weight functions such that $\max_{i=1}^k{(n \cdot w_i(e))} \leq B$.
\section{Dynamic isolation from static non-zero circulation}\label{sec:iso}
One plausible approach to maintain distance or matching in a graph is to give
dynamically evolving non-zero circulations. Even for planar graphs, this is
a tough order because giving non-zero circulations requires a 
dynamic planar embedding algorithm in $\DynFO$, which is not known. Further,
small changes in the input can lead to large changes in the embedding.
Thus, any algorithm that gives non-zero circulation to planar graphs like
\cite{BTV09,TV12} and uses Logspace algorithms like undirected reachability
\cite{Reingold08} cannot be used even if it was made dynamic.
This induces us to side-step maintaining non-zero circulation weights.

Instead, we use a result from \cite{FGT} to convert the given static circulation
to dynamic isolating weights. Notice that \cite{FGT} yields a black box recipe
to produce isolating weights of quasipolynomial magnitude, which we label
as \emph{FGT-weights} in the following way.
Given a bipartite graph $G$, one first considers a non-zero circulation
of exponential magnitude viz. $w_0(e) : e \mapsto 2^e$. Next, consider a 
list of $\ell = O(\log{n})$ primes $\vec{p} = (p_1,\ldots,p_{\ell})$ 
which yield a weight function $w_{\vec{p}}(e)$. This is defined
by taking the $\ell$ weight functions $w_0\bmod{p_i}$ for $i \in \{1,\ldots,\ell\}$
and concatenating them with \emph{shifting} the weights to the higher-order bits appropriately, that is:
$w_{\vec{p}}(e) = \left<w_0(e) \bmod{p_1},\ldots,w_0(e) \bmod{p_\ell}\right>$.
This is so that there is no overflow from the $i$-th field to the 
$(i-1)$-th for any $i \in \{2,\ldots,\ell\}$. 

Suppose we start with a graph with static weights ensuring non-zero 
circulation. In a step, some edges are inserted or deleted. 
The graph after deletion is a subgraph of the original graph; hence the
non-zero circulation remains non-zero after a deletion\footnote{If we merely had isolating
weights, this would not necessarily preserve isolation.}, but we have to do more
in the case of insertions. We aim to give the newly inserted edges FGT-weights
in the higher-order bits while giving weight $0$ to all the original edges in 
$G$ again in the higher-order bits. Thus the weight of all perfect matchings
that survive the deletions in a step remains unchanged. Moreover, if none such
survive but new perfect matchings are introduced (due to insertion of edges)
the lightest of them is determined solely by the weights of the newly
introduced edges. In this case, our modification of the existential proof 
from \cite{FGT} ensures that the minimum weight perfect matching is unique.

Notice that the FGT-recipe applied to a graph with polylog ($N = \log^{O(1)}{n}$) edges yields quasipolylogarithmically ($N^{\log^{O(1)}{N}}$) large weights which are therefore still subpolynomial 
($2^{(\log{\log{n}})^{O(1)}} = 2^{o(\log{n})} = n^{o(1)}$).
Thus the weights remain polynomially bounded when shifted to accommodate for the old weights. Further the number of primes is polyloglog 
($\log^{O(1)}{N} = (\log{\log{n}})^{O(1)}$) 
and so sublogarithmic ($\log^{o(1)}{n}$).
Hence, the number of possible different weights is subpolynomial, 
which allows us to derandomize our algorithm.

Before getting into technical details, we point out that in \cite{DKMTVZ20} a
similar scheme is used for reachability and bears the same relation to
\cite{KT16} as this section does to \cite{FGT}. 
We have the following lemma, which we prove in Appendix~\ref{subsec:isoApp}:
\begin{lemma} \label{lem:combFGT}
Let $G$ be a bipartite graph with a non-zero circulation
$w$. Suppose $N = \log^{O(1)}{n}$ edges are inserted into $G$ to yield $G^{new}$
then we can compute polynomially many weight functions in $\FOar$ that
have $O(\log{n})$ bit weights, and at least one of them,
$w^{new}$ is isolating. Furthermore, the weights of the original edges remain
unchanged under $w^{new}$.
\end{lemma}

\section{Maximum Cardinality Matching Search in $\DynFOpar$}\label{sec:ac0p}
In this section, we convert the static algorithm for maximum matching search in bipartite graphs into a dynamic algorithm with the help of the isolating weights from the previous section. In the static setting \cite{DKKM18} the problem reduces to determining non-singularity of an associated matrix given a non-zero circulation for the graph. 

The algorithm
extracts what is called a min-weight \emph{generalized} perfect matching 
(min-weight GPM), that is,  a matching along with some self-loops. 
The construction proceeds by adding a distinct edge $(v,t_v)$ 
on every vertex $v \in V(G)$ with a self-loop on the new vertex $t_v$
to yield the graph $G'$.
 The idea is to match as many vertices as possible in $G'$ using the 
actual edges of $G$ while reserving the pendant edges $(v,t_v)$
to match vertices that are
unmatched by the maximum matching. If a vertex $v$ is matched in a 
maximum matching of $G$ then the vertex $t_v$ is ``matched'' 
using the self-loop. 

Given a non-zero circulation weight $w'''$ for $G$ 
the weight function for $G'$ is $w = \left<w',w'',w'''\right>$. 
Here we represent by $w'$ the function that is identically
$0$ for all the self loops and is $1$ for all the other edges.
$w''(e)$ is zero except for pendant edges $e = (v,t_v)$, for $v \in V(G)$,
which have
$w''(e) = v$ (where $v$ is interpreted as a number in $\{1,\ldots,|V(G)|\}$)
such that all vertices get distinct numbers. The paper \cite{DKKM18} 
considers the weighted Tutte matrix $T$ where for an edge $(u,v)$
the entry $T(u,v) = \pm x^{w(u,v)}$ (with a positive sign say iff $u < v$)
and is zero otherwise. It shows that in the univariate polynomial $det(T)$
the least degree term $x^W$ with a non-zero coefficient must have
this coefficient equal to $\pm 1$ and the exponent $W$ is twice the weight
of the minimum weight maximum matching. The edges in the maximum 
matching can then be obtained by checking if on removing the edge $(u,v)$ 
the weight of the minimum weight maximum matching changes.

The idea behind the proof (see Appendix~\ref{subsec:isaacProofApp} for a
detailed version) is that:
\begin{enumerate}
\item The most significant weight function $w'$ ensures that the cardinality 
of the actual edges (i.e. edges from $G$) picked in the GPM in $G'$ equals the 
cardinality of the maximum matching in $G$. 
\item The next most significant weight function $w''$ is used to ensure 
that all the GPMs use the same set of pendant edges. 
\item The least significant weight function $w'''$ then isolates the GPM since
all min-weight GPM's are essentially perfect matchings restricted to the same
set $S$ of vertices, namely those that are not matched by the corresponding 
pendant edges and the non-zero circulation weights on $G$ ensures that 
these are isolating weights on the induced graph $G[S]$.
\end{enumerate}
We would like to claim that we just need \emph{isolating}
 weights $w'''$ instead of
non-zero circulation weights to ensure that the above technique works.
The first two steps go through unchanged. However, the third step does not work
since isolating weights for $G$ might not be isolating weights for the subgraph
$G[S]$. However, Lemma~\ref{lem:combFGT} can be applied to the graph $G[S]$
directly -- notice that in the above proof sketch $S$ is determined by the first
two weight functions $w',w''$ and does not depend on the third $w'''$.

As described above, we need to maintain the determinant of a certain matrix
related to the Tutte matrix in order to find the size of the maximum
cardinality matching.
The Matrix Determinant Lemma 
allows us to maintain the determinant by reducing it to maintaining the inverse
of the matrix. To maintain the inverse, the Sherman-Morrison-Woodbury formula tells us how to reduce the task of recomputing the inverse of a non-singular matrix under small changes to that of computing the inverse of a small matrix statically. So we need to ensure that the matrix remains invertible throughout which is what we achieve below by tinkering with the definition of the Tutte matrix.

We define the following generalized Tutte matrix with rows and columns 
indexed by $V(G) \cup \{t_v : v \in V(G)\}$ with the following weights:
$T(t_v,t_v) = 1$, $T(v,v) = x^{w_\infty}$ and 
$T(a,b) = x^{\left<w',w'',w'''\right>(a,b)}$ whenever $a \neq b$. 
Here $w_\infty$ is a polynomially bounded number larger than the largest
of the weights $\left<w',w'',w'''\right>(a,b)$.
We now have the following:
\begin{lemma}\label{lem:genTutte} Let $T$ be the matrix defined above, then 
the highest exponent $w$ such that $x^w$ divides $det(T)$ is the weight of the
min-weight maximum cardinality matching in $G + w_\infty \times$
the number of unmatched vertices in the maximum matching. Further, the matrix
$T$ is invertible.
\end{lemma}
\begin{proof}
From the properties of the weight function $\left<w',w'',w'''\right>$ 
inherited from \cite{DKKM18} we see that the isolated minimum weight 
generalized perfect matching corresponds to the monomial with the least exponent
as described in Lemma~\ref{lem:genTutte}. 
In order to guarantee invertibility of $T$ we just need to prove that the 
product of the diagonal terms
yields a monomial of higher degree than any other monomial and of coefficient
one, since this implies that the matrix is non-singular.
Consider the product of the
diagonal entries viz. $x^{\frac{|V(G')|}{2}w_{\infty}} = x^{|V(G)|w_{\infty}}$.
Any monomial with less than $|V(G)|$
diagonal entries $T(v,v)$ is bound to be of much smaller exponent. Now the
monomials which use an off-diagonal entry $T(v,u)$ or $T(v,t_v)$
must miss out on the diagonal entry in the $v$-th row, making the exponent 
much smaller.
\end{proof}

\subsection{Maintaining the determinant and inverse of a matrix}
We need the following definitions and results about univariate polynomials,
matrices of univariate polynomials and operations therein over a finite field
of characteristic $2$. Let $\Fbb_2$ be the field of characteristic
$2$ containing $2$ elements.
For (potentially infinite) power series $f,g \in \Fbb_2[[x]]$,
we say $f$ \emph{$m$-approximates} $g$ (denoted by
$f \approx_m g$) if the first $m$ terms of $f$ and $g$ are the same.
We will extend this notation to matrices and write $F \approx_m G$
where $F,G \in \Fbb_2[[x]]^{\ell\times\ell}$ are matrices of power series.
We will have occasion to use this notation only when one of $F,G$ is
a matrix of polynomials, that is, a matrix of finite power series.

Consider the Matrix Determinant Lemma:
$$\det(A + UBV) =  \det(I + BVA^{-1}U) \det(A)$$
and the Sherman-Morrison-Woodbury formula:
$$(A + UBV)^{-1} = A^{-1} - A^{-1}U(I + BVA^{-1}U)^{-1}BVA^{-1}$$
Notice that if $A \in \Fbb_2[x]^{n\times n}$ with the degree of entries
bounded by $w_{\infty}$, then there exists $A^{-1} \in \Fbb_2[[x]]^{n\times n}$.
For us, only the monomials with degrees at most $w_\infty$ are relevant.
Thus we will assume that we truncate $A^{-1}$ at $w_\infty$ many terms
to yield matrix $A' \approx_{w_\infty} A^{-1}$. Then we have the following:

\begin{lemma}\cite[Lemma 10]{DKMTVZ20}\label{lem:approxWoodbury}
Suppose $A \in \Fbb_2[x]^{n \times n}$ is invertible over $\Fbb_2[[x]]$, 
and $C \in \Fbb_2[x]^{n \times n}$ is an $m$-approximation of $A^{-1}$. 
If $A + UBV$ is invertible over $\Fbb_2[[x]]$ with 
$U \in \Fbb_2[x]^{n \times \ell}, B \in \Fbb_2[x]^{\ell \times \ell},$ and 
$V \in \Fbb_2[x]^{\ell \times n}$, then
      $(A + UBV)^{-1} \approx_m C-CU(I+BVCU)^{-1}BVC$. Furthermore, if $\ell \leq \log^c n$ for some fixed $c$ and all involved 
polynomials have polynomial degree in $n$, then the right-hand side can be 
computed in $\FOpar$ from $C$ and $\Delta A$.
\end{lemma}
Similar to the above (using the closure of $m$-approximation under
product), we get an approximate version of the Matrix Determinant Lemma, that is:
\begin{proposition}\label{prop:approxMatDet}
Suppose $A \in \Fbb_2[x]^{n \times n}$ is invertible over $\Fbb_2[[x]]$, 
and $C \in \Fbb_2[x]^{n \times n}$ is an $m$-approximation of $A^{-1}$ 
and polynomial $d(x)\approx_m det(A)$ then 
$d\cdot\det(I + BVCU)\approx_m \det(A + UBV)$.
\end{proposition} 
We can now proceed to prove Theorem~\ref{thm:dynFO2}:
\begin{proof}[Proof of Theorem \ref{thm:dynFO2}]
By putting $m = w_\infty$ and applying
Lemma~\ref{lem:approxWoodbury} and
Propositions~\ref{prop:approxMatDet} 
to the generalized Tutte matrix from Lemma~\ref{lem:genTutte}
and using the (Muddling) Lemma~\ref{lem:muddle}, we complete the matching part
of Theorem~\ref{thm:dynFO2} (see Appendix~\ref{sec:distance} for the proof
involving distance).
\end{proof}

\section{Maximum Cardinality Matching in $\DynFOar$}\label{sec:matchFO}
In this section we give a $\DynFOar$ algorithm for maintaining the size of a maximum matching under $O(\frac{\log n}{\log \log n})$ changes. The approach in the previous section, as we needed to maintain polynomials of large (polynomial in $n$) degree, has the limitation that it only gives a $\DynFOpar$ bound. 
The main ingredient here is a new algorithm for maintaining the rank of a matrix in $\DynFOar$ under $O(\frac{\log n}{\log \log n})$ changes (Section~\ref{sec:rank}). Our matching algorithm follows the basic approach of the non-uniform \DynFO\ algorithm of \cite{DKMSZ18}. Here, since we use deterministic isolation weights (as opposed to the randomized isolation weights of \cite{MulmuleyVV87}), with some more work, we obtain a \emph{uniform} \DynFOar\ bound under bulk changes. 

The  algorithm of \cite{DKMSZ18} builds on the well-known correspondence between the size of maximum matching
and the rank of the Tutte matrix of the corresponding graph -- if a graph contains a maximum matching of size $m$ then the associated Tutte matrix is of rank $2m$ \cite{L79}. The dynamic rank algorithm from Section~\ref{sec:rank} cannot be applied directly since the entries of the Tutte matrix are indeterminates. However, the rank can be determined by replacing each $x_{ij}$ by $2^{w(i,j)}$. Here $w$ assigns a positive integer weight to every edge $(i, j)$ under which the maximum matching gets unique minimal weight, i.e., it is matching-isolating. Using the Isolation Lemma~\cite{MulmuleyVV87}, it can be shown that the correspondence between the rank and the size of the maximum matching does not change after such a weight transformation \cite{Hoa10, DKMSZ18}. 

Our algorithm diverges from \cite{DKMSZ18} as we need to deterministically compute these isolating weights and also, to somehow maintain those.  Since we do not know how to maintain such weights directly, as in Section~\ref{sec:ac0p}, we convert the static non-zero circulation weights to dynamic isolating weights using the Muddling Lemma~\ref{lem:muddle}. Given a graph $G$, let $B_w$ be its weighted Tutte matrix with each $x_{ij}$ replaced by $2^{w(i,j)}$ for an isolating weight function $w$. Initially, the static non-zero circulation weights provide such weights. Since we are only interested in computing the rank of $B_w$, we do not need to make the initial modifications of adding pendant edges or self-loops to $G$ as before. So the weight function $w$ is just the non-zero circulation weight $\langle w''' \rangle$ here. In the dynamic process, similar to Section~\ref{sec:ac0p}, we use the FGT-weights $w_{new}$ on top for the newly inserted edges.  We have the following:
\begin{lemma}\label{lem:matchtorank}
Given a dynamic algorithm for maintaining the rank of an integer matrix under $k = O(\log^c n)$ changes at each step for some fixed constant $c$, we can maintain the size of the maximum matching in the same complexity class under $O(k)$ changes for the class of graphs where non-zero circulation weights can be computed in \NC. 
\end{lemma} 
\begin{proof}
Given a graph $G$, assume we have an algorithm for computing the non-zero circulation weights $w$ in $\NC^i \subseteq \AC[\frac{\log^i n}{\log \log n}]$ for some fixed integer $i$. Once these weights $w$ are available, rank($B_w$) can be found in \NC$^2$~\cite{ABO} which is contained in \AC$[\frac{\log^2 n}{\log \log n}]$. Since $O(k)$ changes can occur at each step, during this time, total of $O(k \cdot (\frac{\log^i n}{\log \log n} + \frac{\log^2 n}{\log \log n}))$ many new changes accumulate. As $w$ assigns non-zero circulation weights to the edges of $G$, we can assign weight $0$ to the deleted edges and the weights remain isolating. For the newly inserted edges, which are only polylog($n$) many, we compute the polynomially bounded FGT-weights in \AC$^0$ using Lemma~\ref{lem:crucialFGT}. Thanks to Lemma~\ref{lem:muddle}, in $O(\frac{\log^i n}{\log \log n} + \frac{\log^2 n}{\log \log n}))$ many steps we can take care of all the insertions by adding $k$ new edges at each step along with $k$ old ones in double the speed using our rank algorithm. Note that, during the static rank computation phase, we do not restart the static algorithm for computing the weight $w$. Instead, we recompute these weights once the rank computation using them finishes. More precisely, we can think of a combined static procedure that computes the non-zero circulation weights followed by the rank of the weighted Tutte matrix $B_w$ in \NC$^b$ for $b$=max$(i,2)$. And on this combined procedure, we apply our Muddling Lemma~\ref{lem:muddle}. 
\end{proof}
We can now prove Theorem~\ref{thm:matchFO}:
\begin{proof} [Proof of Theorem~\ref{thm:matchFO}]
Similar to \cite[Theorem 16]{DKMSZ18} this implies a \emph{uniform} bounded expansion first-order truth-table (bfo-tt) reduction (for the definition see \cite[Section 3.2]{DKMSZ18} where it was first formalized) from maximum matching to rank in this special case. Since dynamic rank maintenance is in \DynFOar\ under $O(\frac{\log n}{\log \log n})$ changes (see Theorem~\ref{thm:rank}),  in classes of graphs where non-zero circulation weights can be computed in $\NC$ we have the result.
\end{proof}

\section{Maintaining Rank in $\DynFOar$ under bulk changes}\label{sec:rank}
In \cite{DKMSZ18}, it was shown that the rank of a matrix with small integer entries over 
the rationals can be  maintained in $\DynFO$ under changes that affect a single entry at a time. In this section, we generalize this result to changes that affect slightly sublogarithmically many entries ($O(\frac{\log n}{\log \log n})$ to be 
precise). For ease of exposition, we build upon the algorithm as described
in \cite{DKMSZ15}. 
Following is the main theorem of this section:
\begin{theorem}\label{thm:rank}
Rank of a matrix is in $\DynFOar$ under $O(\frac{\log n}{\log \log n})$ changes.
\end{theorem}
Before going into the details of the proof,  we start with defining some important notation, followed by our overall proof strategy. 
Let $A$ be a $n \times n$ matrix over $\mathbb{Z}_p$, where $p = O(n^3)$ is a prime. Let $K$ be the kernel of $A$. For a vector $v \in \mathbb{Z}_p^n$, we define $S(v) = \{ i \in n \mid (Av)_i \neq 0\}$, where $(Av)_i$ denotes the $i$th coordinate of the vector $Av$. Let $B$ be a basis of $Z_p^n$. A vector $v\in B$ is called $i$-\textit{unique} with respect to $A$ and $B$ if $(Av)_i \neq 0$ and $(Aw)_i = 0$ for all other $w \in B$. A basis $B$ is called $A$-good if all the vectors in $B - K$ are $i$-unique with respect to $A$ and $B$. For a vector $v \in B-K$, the minimum $i$ for which it is $i$-unique is called the principal component of $v$, denoted as pc($v$). 

Starting with an $A$-good basis $B$, and introducing a small number of changes
to yield $A'$ may lead to $B$ losing its $A$-goodness. To restore this, we alter the
matrix $B$ in four steps to obtain an $A'$-good basis $B'$. The first step
involves identifying a full rank submatrix in $A'$ corresponding to the changed
entries and inverting it. The second step restores the pc's of the affected
columns. In the third and fourth steps, we restore the pc's of the rest of 
vectors, which had lost the pc's. The rough outline is similar to that of
\cite{DKMSZ15} but in order to handle non-constant changes we have to make 
non-trivial alterations and use efficient small matrix inversion from \cite{DMVZ18}.
We have Theorem~\ref{thm:rank} from the following, whose proof is provided 
towards the end of this section:
\begin{lemma}\label{lem:rank}
Let $A, A'\in \mathbb{Z}_p^{n \times n} $ be two matrices such that $A'$ differs from $A$ in $O(\frac{\log n}{\log \log n})$ places. If $B$ is an $A$-good basis then we can compute an $A'$-good basis $B'$ in $\FOar$.
\end{lemma}
%
\begin{proposition}[\cite{DKMSZ18}]\label{prop:rankAgood}
Let $A\in\mathbb{Z}_p^{n\times n}$ and $B$ an $A$-good basis of 
$\mathbb{Z}_p^n$. Then $rank(A) = n - |B \cap K|$ is the number of
vectors in the basis that have a pc.
\end{proposition}
\begin{claim}\label{clm:prime}
If the rank of an $n \times n$ matrix $A$ is $r$ then there exists a prime $p = O(max(n,\log N)^3)$ such that the rank of $A$ over $\mathbb{Z}_p$ is also $r$, where $N$ is the maximum absolute value the entries of the matrix $A$ contain.
\end{claim}
\begin{claimproof}
We know that if the rank of $A$ is $r$ then there exists a $r \times r$ submatrix $A_s$ of $A$ such that its determinant is nonzero. The value of this determinant is at most $n!N^n$, which can be represented by $O(n(\log n + \log N))$ many bits. Therefore, this determinant is divisible by at most $O(n(\log n + \log N))$ many primes. Thus by the prime number theorem, we can say that for a large enough $n$ there exists a prime $p$ of magnitude $O(max(n,\log N)^3)$ such that determinant of $A_s$ is not divisible by $p$. 
\end{claimproof}
Hence, to compute the rank of $A$, it is sufficient to compute the rank of the matrices $(A$ mod $p)$ for all primes $p$ of size $O(max(n,\log N)^3)$ and take the maximum among them. Below we show how to maintain the rank of the matrix $A\bmod{p}$ for a fixed prime $p$. We replicate the same procedure for all the primes in parallel. 
%
%

$A'$ is the matrix that is obtained by changing $k$ many entries of $A$. Notice that if $B$ is not $A'$-good basis, that means there are some vectors in $B-K'$ which are not $i$-unique with respect to $B$ and $A'$, where $K'$ is the kernel of $A'$. A vector $w \in B-K'$  which was $i$-unique $(i \in [n])$ with respect to $A$ and $B$ may no longer be $i$-unique with respect to $A'$ and $B$ for the following two reasons,
(i) $i \notin S'(w)$,  (ii) there may be more than one vector $w'$ such that $i \in S'(w')$.
For a vector $v$, $S'(v)$ denotes the set of non-zero coordinates of the vector $A'v$. Below we give an $\AC^0$ algorithm 
to construct an $A'$-good basis. 
\subsection{ Construction of an $A'$-good basis}
Let $k = O (\frac{\log n}{\log \log n} )$ and $M =(A'B)_{R,C}$ be the
$n \times n$ matrix where $R$ is the set of rows and $C$ is the set of columns
of $M$.  We know that $A'B$ differs from $AB$ in a set $R_0$ of at
 most $k$ many rows.
Let $M_{R_0,*}$ be the matrix $M$ restricted to the rows in the set $R_0$.
\begin{claim} \label{clm:smallPrime}
 There exists a prime $q = O(\log^3 n)$ such that the rank of $(M_{R_0,*}\bmod{q})$ is equal to the rank of  $(M_{R_0,*} \bmod{p})$.
\end{claim}
\begin{claimproof}
The proof follows from the proof of Claim \ref{clm:prime}. 
\end{claimproof}
From the above claim, it follows that a row basis of $(M_{R_0,*}\bmod{p})$
remains a row basis of $(M_{R_0,*}\bmod{q})$ for some
$O(\log{\log{n}})$-bit prime $q$.
 Next, we have two constructive claims: 
\begin{claim}\label{clm:rowBasis}
A row basis $R_1$ of ($M_{R_0,*}\bmod{q})$ can be found in $\ACz$.
\end{claim}
\begin{claimproof}
Note that number of rows in $R_0$ are $O (\frac{\log n}{\log \log n})$; thus, the number of the subsets of the rows of $R_1$ are polynomially  many (in $n$), and each row in the set $R_0$ can be indexed by $O(\log \log n)$ many bits. An element of $\mathbb{Z}_q$ can also be represented by $O(\log \log n)$ many bits. Therefore, for a fixed subset $S$ of $R_0$, all the linear combinations of the rows of $S$ can be represented by $O(\log n)$ many bits. We try all the linear combinations in parallel. Also, we do this for all the subsets in parallel. The subset with the maximum cardinality in which all the linear combinations result in non-zero values will be the maximum set of linearly independent rows in ($M_{R_0,*}$ mod $q)$.
\end{claimproof}
\begin{claim}\label{clm:colBasis}
A column basis $C_1$ of $(M_{R_1,*}\bmod{q})$ can be found in $\ACz$.
\end{claim}
\begin{claimproof}
To find the maximum set of linearly independent columns in the matrix $M_{R_1,*}$ we just check in parallel if the rank of  $M_{R_1,i}$ is greater than the rank of $M_{R_1,i-1}$ for all $i \in [m]$. Let $c_1,c_2 \ldots c_m$ be the columns in the matrix $M_{R_1,*}$. Note that the set of columns $c_i$ such that the rank of $M_{R_1,i}$ is more than the rank of $M_{R_1,i-1}$, form a maximum set of linearly independent columns in $M_{R_1,i}$. We can check this in $\ACz$.  
\end{claimproof}
%
%
We are going to construct four matrices $D^{(1)},E^{(1)},D^{(2)},E^{(2)}$
successively such that the product 
$B' = B \times D^{(1)}\times E^{(1)}\times D^{(2)}\times E^{(2)}$ is an
$A'$-good basis. For this, we need to show that each column $c_i$ of $A'B'$ is 
either an all zero-column or there exists a unique $j$ such that the $j$-th
entry of the column is non-zero. In other words, each column of $B'$ is either 
$i$-unique or it is in the kernel of $A'$. We will show how to obtain each
of the four matrices above as well as take their product in $\ACz$.
We need a technical lemma before we start.
\subparagraph*{Combining Matrices}
Here we state a lemma about constructing matrices from smaller matrices that we
we will use several times.
Let $X\in\mathbb{Z}_p^{n\times n}$ be a matrix and
let $X^{1,1} \in \mathbb{Z}_p^{\ell\times\ell}$,
$X^{1,2} \in \mathbb{Z}_p^{\ell\times (n-\ell)}$,
$X^{2,1} \in \mathbb{Z}_p^{(n-\ell)\times \ell}$,
$X^{2,2} \in \mathbb{Z}_p^{(n-\ell)\times(n-\ell)}$ be $4$ matrices
and let $R,C \subseteq [n]$ be two subsets of indices of cardinality $\ell$
each. Let $\bar{R} = [n]\setminus R, \bar{C} = [n]\setminus C$. Then we have (see Appendix \ref{subsec:rankApp} for the proof):
\begin{lemma}\label{lem:combineMatrices}
Given the matrices $X^{i,j}$ for $I,J \in [2]$ and the sets $R,C$ explicitly
for $|R| = |C| = \ell = (\log{n})^{O(1)}$,
we can construct, in $\AC^0$,  the matrix $Y$ such that $Y_{R,C} = X^{1,1}$,
$Y_{R,\bar{C}} = X^{1,2}$, $Y_{\bar{R},C} = X^{2,1}$ and
$Y_{\bar{R},\bar{C}} = X^{2,2}$.
\end{lemma}
\subparagraph*{Phase 1}
First, we restore the $i$-uniqueness of the columns indexed by the set $C_1$.
Let $R_{C_1}$ be the set of rows in $R$ indexed by the same set of indices
 as $C_1$ in $C$.
We right multiply $M$ with another matrix
$D^{(1)}\in \mathbb{Z}_p^{n \times n}$ such that $D^{(1)}_{R_{C_1},C_1}$ is the
inverse of $M_{R_1,C_1}$ and $D^{(1)}_{R-R_{C_1},C-C_1}$ is the identity matrix
and all the other entries of $D^{(1)}$ are zero. Since the inverse of a
$k \times k$ matrix can be computed in $\AC^0$ \cite{DKMTVZ20},
matrix $D^{(1)}$ can be obtained in $\ACz$ via Lemma~\ref{lem:combineMatrices}.

Let ${M^{(1)}} = M \times D^{(1)}$, note that
${M^{(1)}}_{R_1,C_1}$ is an identity matrix. Since $M = A' \times B$,
we have
${M^{(1)}} = M \times D^{(1)}= A' \times B \times D^{(1)}.$
Note that since ${M^{(1)}}_{R_1,C_1}$ is an identity matrix,  the vectors corresponding to the columns in $C_1$ in the matrix $B \times D^{(1)}$ can now easily be made $i$-unique. Since ${M^{(1)}}_{R_1,C_1}$ is an identity matrix, all columns in the matrix ${M^{(1)}}_{R_1, C-C_1}$ can be written as the linear combinations of columns of ${M^{(1)}}_{R_1,C_1}$. Let $\tilde{M}^{(1)}$ be a matrix defined as $\tilde{M}^{(1)} = {M^{(1)}} \times E^{(1)}$, where $E^{(1)} \in \mathbb{Z}_p^{n \times n}$ is constructed as follows.
(i) $E^{(1)}_{R_{C_1},*}$ is same as $M^{(1)}_{R_1,*}$.
(ii) $E^{(1)}_{R- R_{C_1},C_1}$ is the zero matrix.
(iii) $E^{(1)}_{R-R_{C_1},C-C_1}$ is the negative identity matrix.
 Using Lemma~\ref{lem:combineMatrices}, we can construct $E^{(1)}$ in $\ACz$.
Note that $\tilde{M}^{(1)}_{R_1,C_1}$ is an identity matrix and
$\tilde{M}^{(1)}_{R_1,C-C_1}$ is a zero matrix. Thus we can say that vectors
corresponding to columns in $C_1$ in the matrix
$B \times D^{(1)} \times  E^{(1)}$ are $i$-unique for some $i\in R_1$.

Next, we perform a procedure similar to Phase 1 for those vectors which lost
their pc's when we changed the matrix from $A$ to $A'$,
i.e. those vectors $w$ which were $i$-unique for some $i$, but $i \notin S'(w)$.\subparagraph*{Phase 2}
 There can be at most $k$ vectors which lost their pc's while changing the
matrix from $A$ to $A'$. Some of these vectors might get their pc's set in
Phase 1.  Let $\tilde{C}_2$ be the remaining set of vectors in $B$.
Notice $C_1 \cap \tilde{C}_2$ is empty.
To set the pc's of these vectors, we repeat the above procedure for the matrix
$\tilde{M}^{(1)}_{*,\tilde{C}_2}$ as follows.

Find the column basis $C_2$ of $\tilde{M}^{(1)}_{*,\tilde{C}_2}$ in $\AC^0$
recalling  that $|\tilde{C}_2| \leq k$ and using Claim~\ref{clm:rowBasis}
on the transpose of $\tilde{M}^{(1)}_{*,\tilde{C}_2}$. By considering
the transpose of $\tilde{M}^{(1)}_{*,{C}_2}$ and applying
Claim~\ref{clm:colBasis} we can get a row basis $R_2$ of
$\tilde{M}^{(1)}_{*,{C}_2}$. Notice that $R_1 \cap R_2$ is empty.
We construct a matrix $D^{(2)}\in \mathbb{Z}_p^{n\times n}$ in $\AC^0$
using Lemma~\ref{lem:combineMatrices} such that
 $D^{(2)}_{{R_{C_2}},{C_2}}$ contains the inverse of
$\tilde{M}^{(1)}_{{R_2},{C_2}}$, $ D^{(2)}_{R -{R_{C_2}},C -{C_2}}$ is
an identity matrix and the rest of the entries of $D^{(2)}$ are zero.
Here $R_{C_2}$ is the set of rows indexed by the same indices as in the
set $C_2$ of columns.
Let $\tilde{M}^{(2)} = \tilde{M}^{(1)} \times D^{(2)} \times E^{(2)}$, where
$E^{(2)}$ is constructed in $\AC^0$ (using Lemma~\ref{lem:combineMatrices})
so that:
(i)  $\tilde{M}^{(2)}_{R_1 \cup {R_2},C_1 \cup {C_2}}$ is the identity matrix.
(ii) $\tilde{M}^{(2)}_{R_1,C-C_1}$ and $\tilde{M}^{(2)}_{R-R_1,\tilde{C}_2-{C_2}}$ are zero matrices.
Finally, we have $B' = B \times D^{(1)} \times E^{(1)} \times D^{(2)} \times E^{(2)}$.
Since each column of the newly constructed matrices
 contains at most $(k+1)$ non-zero entries, we can obtain $B'$ in $\ACz$.
\begin{claim}\label{clm:AmodGood}
$B'$ is an $A'$-good basis.
\end{claim}
\begin{claimproof}
First, we prove that the vectors which lost their pc's, either get new pc's or they are modified to be in the kernel of $A'$. Let $w\in B$ be a vector that lost its pc and it is not captured in both Phase 1 or Phase 2. Assume it is not captured in Phase $1$, i.e. the vector $A'w$ does not belong to the column set indexed by $C_1$. Then it will be captured in Phase 2. If it is not captured in Phase 2 as well, then we can say that $A'w$ does not belong to the columns indexed by the set $C_2$. Therefore it can be written as a linear combination of the vectors in $C_2$. In Phase $2$, we modify such vectors in a way that $A'w$ becomes a zero vector, i.e. $w$ goes into the kernel of $A'$. Also, note that the vector which did not lose their pc's and are not captured in Phase 1 and Phase 2, do not lose their pc's in the procedure.
\end{claimproof}
%
We prove that we can maintain the number of pc's in $B$ in $\ACz$
using the next claim. However, we need to set up some notation first.
Let $P^{\text{old}}, P^{\text{new}}$ be respectively, the number of pc's
before and after a round.
Let $V_{R_1}^{\text{new}}$ and $V_{R_2}^{\text{new}}$ denote the set of
vectors that have their pc's in the rows $R_1$ and $R_2$, after the round.
Let $V_{R_0}^{\text{old}}$ denote the set of vectors that have their pc's in the
rows $R_0$ before starting of the round and $V_1$ denotes the set of vectors
 which have their pc's in the rows $R-R_0$ before the round and attain pc's in
the rows $R_1$ after the round. Note that all the cardinalities of all the
sets of vectors mentioned above are $O( \frac{\log n}{\log \log n})$.
Therefore, we can compute their cardinalities in $\ACz$.
\begin{claim}\label{clm:numPCs}
$P^{\text{new}}= P^{\text{old}} - |V_{R_0}^{\text{old}}| + |V_{R_1}^{\text{new}}| +|V_{R_2}^{\text{new}}| - |V_1|$.
\end{claim}
\begin{claimproof}
First, we assume that all the vectors in the set $V_{R_0}^{\text{old}}$ lose their pc's after the round therefore we subtract $|V_{R_0}^{\text{old}}|$ from $P$. But some of these vectors get their pc's in the phase $1$ and phase $2$. Therefore, we add $|V_{R_1}^{\text{new}}|$ and $|V_{R_2}^{\text{new}}|$ back to the sum. Notice that $V_{R_1}^{\text{new}}$ may contain those vectors as well that had their pc's in the rows indexed by $R-R_0$ before the round. That means these vectors had a pc before and after the round, but we added their number $|V_{R_1}^{\text{new}}|$. Therefore, we subtract the number of such vectors by subtracting $|V_1|$ from the total sum.
\end{claimproof}
This brings us to the proof of Lemma~\ref{lem:rank}. 
\begin{proof}[Proof of Lemma~\ref{lem:rank}]
The proof is complete from the above claims because the number of pc's is
precisely the rank of the matrix as a consequence of Proposition~\ref{prop:rankAgood}.
\end{proof}
\section{Conclusion}\label{sec:concl} 
In this work, we prove two meta-theorems for distance and maximum matching, which provide the best known dynamic bounds in graphs where non-zero circulation weights can be computed in parallel. This includes important graph classes like planar, bounded genus, bounded treewidth graphs. 
We show how to non-trivially modify two known techniques and combine them with existing tools to yield the best known dynamic bounds for more general classes of graphs, and at the same time allow for bulk updates of larger cardinality. While for bipartite matching we are able to show a $\DynFOar$ bound it would be interesting to achieve this also for maintaining distances, even in planar graphs under single edge changes.

\bibliography{references}

\appendix
\section{Preliminary Appendix}
\label{sec:prelimsApp}
\subsection{Abbreviations}\label{subsec:langDefs}
\begin{definition}
\begin{itemize}
\item $\PM$ is a set of independent edges covering all the vertices of a graph. 
\item $\BPM$ is a $\PM$ in a bipartite graph.
\item $\PMD$ Given an undirected graph $G$ determine if there is a perfect matching in $G$.
\item $\PMS$ Construct a perfect matching for a given graph $G$ if one exists.
\item $\BMWPMS$ Given an edge weighted bipartite graph construct a $\PM$ of least weight.
\item $\MCM$ is the maximum set of independent edges in a given graph.
\item $\BMCM$ is an $\MCM$ in a bipartite graph.
\item $\MWMCM$ is an $\MCM$ of least weight in an edge weighted graph.
\item $\BMWMCM$ is an $\MWMCM$ in a weighted bipartite graph.
\item $\MCMSz$ determine the size of an $\MCM$.
\item $\BMCMSz$ the $\MCMSz$ problem in bipartite graphs.
\item $\BMWMCMS$ Construct a $\BMWMCM$. 
\item $\Reach$ Given a directed graph $G$ and two vertices $s$ and $t$,  is there a path from $s$ to $t$ in $G$.
\item $\Dist$ Given a directed graph $G$ with polynomially bounded edge weights and two vertices $s$ and $t$,  find the weight of a least weight path from $s$ to $t$.
\item $\Rank$ Given an $m \times n$ matrix $A$ with integer entries find the rank of $A$ over $\mathbb{Q}$.
\end{itemize}
\end{definition}
\subsection{Preliminaries on dynamic complexity}\label{sec:dynprelims}
The goal of a dynamic program is to answer a given query on an \emph{input structure} subjected to changes that insert or delete tuples. The program may use an auxiliary data structure represented by an \emph{auxiliary structure} over the same domain. Initially, both input and auxiliary structures are empty; and the domain is fixed during each run of the program. We say a query is in $\DynFO$ if after each insertion or deletion of a tuple, we can update the auxiliary structure in \FO\ or equivalently in Dlogtime-uniform \ACz.

There are several roughly equivalent ways to view the complexity class $\DynFO$
as capturing:
\begin{itemize}
\item The dynamic complexity of maintaining a Pure \textsc{SQL} database under 
fixed (first-order) updates and queries (the original formulation from \cite{PatnaikI}).
\item The circuit dynamic complexity of maintaining a property where the 
updates and queries use uniform $\ACz$ circuits (see \cite{BIS} for the equivalence of uniform $\ACz$ and $\FO$).
\item The parallel dynamic complexity of maintaining a property where the 
updates and queries use constant time on a CRCW PRAM (for the definition 
of Concurrent RAM, see \cite{Immerman}).
\end{itemize}
The first characterization is popular in the Logic and Database community, 
while the second is common in more complexity-theoretic contexts.
The third one is useful to compare and contrast this class with dynamic 
algorithms, which essentially classify dynamic problems in terms of the
sequential time for updates and queries. 
Operationally, our procedure is easiest to view in terms of 
the second or even the third viewpoint.
We would like to emphasize that 
modulo finer variations based on built-in predicates (like arithmetic and
order) in the first variation, uniformity in the second one and built-in 
predicates (like shift) in the third one, the three viewpoints are entirely
equivalent. 
%
\subparagraph*{Maintaining Witnesses:} The proof of \cite{MulmuleyVV87}, reducing the construction of a perfect matching witness, carries over to the dynamic setting also and allow us to maintain a witness to the solution in $\DynFOpar$.  Since the perfect matching is isolated, from its weight we can infer the edges in the matching -- just delete the edges one a time in parallel and see if the weight remains unchanged and accordingly place the edge in the matching. This is doable in $\FOpar$. The extraction procedure for shortest path and maximum cardinality matching is similar.
\section{Dynamic Distance under bulk changes}\label{sec:distance}
In this section, we consider the dynamic complexity of maintaining distances in graphs with polynomially bounded edge weights. Though we have a \DynFO\ algorithm for maintaining reachability in general directed graphs, such an algorithm for maintaining distances is not known. The \DynFO\ update algorithm with \TC$^0$-query of Datta et al.~\cite{DMVZ18} remains the best known uniform bound for general graphs, which works even under $O(\log n/\log \log n)$ changes. Here we show that extending the reachability result of~\cite{DKMTVZ20}, distances can be maintained in \DynFOpar\ under polylog$(n)$ changes in classes of graphs where non-zero circulation can be computed in \NC. In the following, we start with describing the reachability algorithm of~\cite{DKMTVZ20} followed by the necessary modifications needed for maintaining the  distance information.  

\subsection{Outline of the approach for Reachability}
Suppose we are given a graph $G = (V,E)$ with $n$ nodes with an isolating weight assignment $w$. For a formal variable $x$, let the corresponding weighted adjacency matrix $A = A_{(G, w)}(x)$ be defined as follows: if $(u,v) \in E$, then $A[u,v] = x^{w(u,v)}$, and $0$ otherwise. Consider the matrix $D = (I - A)^{-1}$, where $I$ is the identity matrix. Notice that the matrix $(I - A)$ is invertible over the ring of formal power series (see \cite{DMVZ18}). Here $D = \sum_{i=0}^\infty (A)^i$ is a matrix of formal power series in $x$ and in the $(s, t)$-entry, the coefficient of the $i$-th terms gives the number of paths from $s$ to $t$ with weight $i$. 

As $w$ isolates the minimal paths in $G$, it is enough to compute these coefficients modulo $2$ for all $i$ up to some polynomial bound since there is a unique path with the minimal weight if one exists. So, it is enough to compute and update the inverse of the matrix $I - A$. Though to do it effectively, we compute the \emph{$n$-approximation} $C$ of $D$, which is a matrix of formal polynomials that agrees with the entries of $D$ up to the terms of degree $i \leq n$. This precision is preserved by the matrix operations we use, see \cite[Proposition~14]{DMVZ18}. 

When applying a change $\Delta G$ to $G$ that affects $k$ nodes, the associated matrix $A$ is updated by adding a suitable change matrix $\Delta A$ with at most $k$ non-zero rows and columns, and can therefore be decomposed into a product $UBV$ of suitable matrices $U, B$ and $V$, where $B$ is a $k \times k$ matrix. To update the inverse, we employ the Sherman-Morrison-Woodbury identity (cf.~\cite{HendersonS81}), which gives a way to update the inverse when the matrix $A$ is changed to $A + \Delta A$ as follows:
\[(A + \Delta A)^{-1} = (A + UBV)^{-1} = A^{-1}-A^{-1}U(I+BVA^{-1}U)^{-1}BVA^{-1}.\]

The right-hand side can be computed in $\FOpar$ for $k = O(\log^c n)$ since modulo $2$ computation of (1) multiplication and iterated addition of polynomials over $\mathbb{Z}$ and (2) computation of the inverse of $I+BVA^{-1}U$ which is also a $k \times k$ matrix is possible in $\FOpar$ for (matrices of) polynomials with polynomial degree~\cite{HealyV}. Finally, we need to assign weights to the changed edges as well so that the resulting weight assignment remains isolating. We show how to achieve this starting with non-zero circulation weights. Using~\cite[Theorem~5]{DKMTVZ20} we can assume that such a weight assignment is given, and that we only need to update the weights once.

Let $u$ be skew-symmetric non-zero circulation weights for $G$ and let $n^k$ be the polynomial bound on the weights. Further, let $w$ be the isolating weight assignment that gives weight $n^{k+2} + u(e)$ to each edge $e \in E$. During the \AC$^d$ initialization, we compute the weights $u$ and $w$ and an $n^b$-approximation matrix $C$ of $(I - A_{(G, w)}(x))^{-1} \bmod 2$ for some constant $b$.  

When changing $G$ via a change $\Delta E$ with deletions $E^-$ and insertions $E^+$, the algorithm proceeds as follows: To compute the isolating weights $w^-$, the non-zero circulation weights $u^-$ for $G^-$ are obtained from $u$ by setting the weight of deleted edges $e \in E^-$ to $0$. As $u^-$ gives the same weight to all simple cycles in $G^-$ as $u$ gives to these cycles in $G$, it has non-zero circulation. To handle $E^+$ it can be shown that~\cite[Lemma~11]{DKMTVZ20} there is a \FO-computable (from $w, E^+$ and the reachability information in $G$) family $W'$ of polynomially many weight assignments such that $\exists w' \in W'$ isolating for $(V, E \cup E^+)$.

Hence we need to maintain polynomially many different instances of the graph with different weight functions from $W'$ such that in at least one of them the paths are isolated. The idea is that if there is an $s$-$t$-path using at least one inserted edge from $E^+$, then there is a unique minimal path among all $s$-$t$-paths that use at least one such edge, while ignoring the weight of the paths that is contributed by edges from $E$. The edge weights from $E^+$ are multiplied by a large polynomial to ensure that the combined weight assignment with the existing weights for edges in $E$ remains isolating. Since the weights are constructed only for a graph with $N = O(\log^c n)$ many nodes, and although they are not polynomially bounded in $N$, they are in $n$. Please refer to \cite[Section 6]{DKMTVZ20} for more details.

From the above discussion, to prove a similar bound for distances, it suffices to show that (1) after every polylog$(n)$ changes, we can ensure the edge weights remain ``shortest path-isolating'' and (2) under such weights the distance can be updated in $\FOpar$.

\subsection{Dynamic Isolation of Shortest Paths}\label{subsec:dyndist}
In the following, we first describe how the isolating weights for reachability can be modified to give weights for isolating shortest paths. Similar to maintaining reachability, our algorithm handles deletions and insertions differently. In the case of deletion, we set the weight of the deleted edges $e \in E^-$ to $0$ and due to the non-zero circulation weights, the weights remain isolating. In the case of insertions, the idea is to do a weight refinement by shifting the original edge weights $w(e)$ ($1$ in case of unweighted graphs) to the highest order bits in the bit-representation, in the presence of other newly assigned weights to the edges. 

We define a new weight functions $w^* = \langle w, w', u\rangle$ and assign these weights to the inserted edges $e \in E^+$. The existing edges are not assigned any $w'$ weight and all those bits remain zeroes. So we get a family of weight functions $W^*$. Here $w$ is the polynomially-bounded original edge weights, $w'$ is one of the (polynomially many) isolating weights from the family $W'$ assigned to the newly added edges during the dynamic process and $u$ are the non-zero circulation weights, which are computed statically. The combined weights $w^*$ is \FO-constructible from the weights $w, w'$ and $u$ as all involving numbers are $O(\log{n})$ bits (see \cite[Theorem~5.1]{HAB02}). The correctness of the fact that these weights are indeed shortest path isolating follows from~\cite[Lemma~11]{DKMTVZ20} with the observation that since the original edge weights are shifted to the highest-order bits, the minimum weight path with these weights corresponds to the shortest path in the original graph.

The update algorithm for maintaining reachability can be extended to maintaining distances also~\cite{DMVZ18}. Here, instead of checking only the non-zeroness of the $(s,t)$-entry in the polynomial matrix $C$, we compute the minimum degree term as well (with coefficient $1$). By construction, the degree of a term in this polynomial is same as the weight of the corresponds path under the dynamic isolating weights and applying an easy transformation gives us back the original weights, that is, the weight of the shortest path from $s$ to $t$ in $G'$. This proves the Distance result of Theorem~\ref{thm:dynFO2}.
\section{The Details from Section~\ref{sec:iso}}\label{subsec:isoApp}
We use the same general strategy as in \cite{DKMTVZ20} and divide the edges into
\emph{real} and \emph{fictitious}, where the former represents the newly inserted
edges and the latter original undeleted edges\footnote{We use the terms
old $\leftrightarrow$ fictitious and new $\leftrightarrow$ real interchangeably
in this section.}.

Let $\calC$ be a set of cycles containing both real and fictitious edges that 
occur in any PM. 
Let $w$ be a  weight 
function on the edges that gives non-zero weight only to the real edges.
Define $c_w(\calC)$ to be the set of all circulations in the cycles of 
$\calC$. Here the circulation is the absolute value of the alternating sum
of weights $c_w(C) = |w(f_1) - w(f_2) + w(f_3) - \ldots|$ where 
$C = f_1,f_2,f_3,\ldots$ is the sequence of edges in the cycle. 

We say that a weight function
that gives non-zero weights to the real edges, \emph{real isolates} $\mathcal{M}$ 
for a set system $\mathcal{M}$ if the minimum weight set in $\mathcal{M}$ is unique. In our context, $\mathcal{M}$ will refer to the set of perfect/maximum
 matchings.

Next, we follow the proof idea of \cite{FGT} but focus on assigning weights
to real edges which are, say, $N$ in number.
 We do this in $\log{N}$ stages starting with a graph $G_0 = G$
and ending with the acyclic graph $G_\ell$ where $\ell = \log{N}$. The inductive
assumption is that:
\begin{invariant}\label{inv:analogFGT}
 For $i \geq 1$, $G_i$ contains no cycles with at most $2^{i+1}$ real edges. 
\end{invariant}

Notice that induction starts at $i > 0$.

We first show how to construct $G_{i+1}$ from $G_i$ such that if $G_i$ satisfies
the inductive invariant~\ref{inv:analogFGT}, then so does $G_{i+1}$.
Let $i > 1$, then in the $i$-th stage, let $\calC_{i+1}$ be the set of cycles that contain 
at most $2^{i+2}$ real edges. For each such cycle 
$C = f_0,f_1,\ldots$ containing $k \leq 2^{i+2}$ real edges (with $f_0$ being the
least numbered real edge in the cycle) edge-partition it into $4$ consecutive 
paths $P_j(C)$ for $j \in \{0,1,2,3\}$ such that the first $3$ paths contain exactly 
$\lfloor\frac{k}{4}\rfloor$ real edges and the last path contains the rest. In
addition ensure that the first edge in each path is a real edge. Let the first
edge of the $4$-paths be respectively $f_0 = f'_0, f'_1, f'_2, f'_3$. We 
have the following which shows that the associated 
$4$-tuples $\left<f'_0,f'_1,f'_2,f'_3\right>$
uniquely characterise cycles in $\calC_{i+1}$.
\begin{claim}
There is at most one cycle in $\calC_{i+1}$ that has a given $4$-tuple 
$\left<f'_0,f'_1,f'_2,f'_3\right>$ associated with it.
\end{claim}
\begin{claimproof}
Suppose two distinct cycles $C,C' \in \calC_{i+1}$ have
a tuple $\left<f'_0,f'_1,f'_2,f'_3\right>$ 
associated with them. Then for least one $j \in \{0,1,2,3\}$
$P_j(C) \neq P_j(C')$. Hence, $P_j(C) \cup P_j(C')$
is a closed walk in $G_i$ containing at most 
$2\times \lceil\frac{2^{i+2}}{4}\rceil = 2^{i+1}$ many real edges,
contradicting the assumption on $G_i$.
\end{claimproof}

This lemma shows that there are at most $N^4$ elements in $\calC_i$.
Next, consider the following lemma from \cite{FKS}:
\begin{lemma}[\cite{FKS}]
For every constant $c>0$ there is a constant $c_0>0$ such that for
every set $S$ of $m$ bit integers with $|S| \leq m^c$,
the following holds: There is a $ c_0 \log{m}$  bit prime
number $p$ such that for any $x,y \in S$ it holds that if $x \neq y$ then 
$x \not\equiv y \bmod{p}$.
\end{lemma}
We apply it to the set $c_{w_0}(\calC_i) = \{c_{w_0}(C) : C \in \calC_i\}$. 
Here, the weight
function $w_0$ assigns weights $w_0(e_j) = 2^j$ to the real edges which are
$e_1,e_2,\ldots,e_N$ in an arbitrary but fixed order.
Notice that from the above claim, the size of this set 
$|w_0(\calC_i)| \leq N^4$. 
And $w_0(e_j)$ is $j$-bits long; hence $c_{w_0}(C)$ for any cycle
$C \in \calC_i$ that has less than $2^{i+2}$ real edges 
is at most $i+j+2 < 4N$-bits long.  Thus, we obtain a prime $p_{i+1}$ 
of length at most $c_0\log{4N}$ by picking $c = 4$. We define 
$w_{i+1}(e_j) = w_0(e_j) \bmod{p_{i+1}}$.

Now consider the following crucial lemma from \cite{FGT}:
\begin{lemma}[\cite{FGT}]\label{lem:crucialFGT}
Let $G = (V,E)$ be a bipartite graph with weight function $w$. 
Let $C$ be a cycle in $G$ such that $c_w(C) \neq 0$. 
Let $E_1$ be the union of all minimum weight perfect matchings in $G$.
Then the graph $G_1(V,E_1)$ does not contain the cycle $C$.
Moreover, all the perfect matchings in $G_1$ have the same weight.
\end{lemma}
Let $B$ be a large enough constant (though bounded by a polynomial in $N$) 
to be specified later.
We shift the original accumulated weight function
$W_i$ and add the new weight function $w_{i+1}$ to obtain:
$W_{i+1}(e) = W_{i}(e)B + w_{i+1}(e)$. 
Apply $W_{i+1}$ on the graph $G_i$ to obtain the graph $G_{i+1}$.
 Inductively suppose we have
 the invariant~\ref{inv:analogFGT}
 that  the graph $G_i$ did not have any cycles containing at least 
$2^{i+1}$ real edges. This property is preserved when we take all the
perfect matchings in $G_i$ and apply $W_{i+1}$ yielding $G_{i+1}$. Moreover, 
from Lemma~\ref{lem:crucialFGT} and the construction of $w_{i+1}$, the cycles of
$\calC_i$ disappear from $G_{i+1}$ restoring the invariant. This yields a weight
function $W_\ell$ using that $\ell = \log{n}$ (see the discussion before
invariant~\ref{inv:analogFGT}).

Notice that it suffices to take $B$ greater than the number of real edges
times the maximum of $w_i(e)$ over $i,e$.
Showing that $G_1$ contains no cycle of length at most $4$ mimics the above
more general proof, and we skip it here.
We can now complete the proof of Lemma~\ref{lem:combFGT}:
\begin{lemma*} (Lemma~\ref{lem:combFGT} restated)
Let $G$ be a bipartite graph with a non-zero circulation
$w^{old}$. Suppose $N = \log^{O(1)}{n}$ edges are inserted into $G$ to yield $G^{new}$
then we can compute polynomially many weight functions in $\FOar$ that
have $O(\log{n})$ bit weights, and at least one of them,
$w^{new}$ is isolating. Further, the weights of the original edges remain
unchanged under $w^{new}$.
\end{lemma*}

\begin{proof}[Proof of Lemma~\ref{lem:combFGT}]
From the invariant above $G_\ell$ does not contain
any cycles. From the construction of $G_\ell$, if $G$ has a perfect matching,
 then
so does $G_\ell$ and hence it is a perfect matching. Notice that $W_\ell$
is obtained from $p_1,\ldots,p_\ell$ that include $O((\log{\log{n}})^2) = o(\log{n})$ 
many bits. Thus there are (sub)polynomially many such weighting functions $W_\ell$, depending on the primes $\vec{p}$.
Let $w = B\cdot W_\ell + w^{old}$ where we recall that $W_\ell(e)$ is non-zero only
for the new (real) edges and $w^{old}$ is non-zero only for the old (fictitious)
edges. Thus, any perfect matching that consists of only old edges is lighter
than any perfect matching containing at least one new edge. Moreover, if 
the real edges in two matchings differ, then from the construction of 
$W_\ell$ (for some choice of $\vec{p}$) both matchings cannot be lightest 
as $W_\ell$ real isolates a matching. Thus the only remaining case is
that we have two distinct lightest perfect matchings, which differ only in the
old edges. But the symmetric difference of any two such perfect matchings
is a collection of cycles consisting of old edges. But each cycle has a
non-zero circulation in the old graph and so we can obtain a matching
of even lesser weight by replacing the edges of one of the matchings in one
cycle by the edges of the other one. This contradicts that both matchings were
of least weight. This completes the proof.
\end{proof}
\subsection{Maintaining the Determinant and Inverse of a Matrix}
We will need the following results derived 
explicitly or implicitly from \cite{HealyV}:
\begin{lemma}[\cite{HealyV}]\label{lem:healyViola} 
Let $f_1,\ldots,f_\ell,g,h \in \Fbb_2[x]$ each of degree bounded by
$d = n^{O(1)}$ where $\ell = \log^{O(1)}{n}$ and let $m = n^{O(1)}$.
Assume the constant term $h(0) = 1$.
Then each of the following is in $\FOpar$:
\begin{enumerate}
\item The product $f_1\cdot f_2\cdot\ldots\cdot f_\ell$
\item The power polynomial $g^m$
\item The (unique) polynomial $\tilde{h} \in \Fbb_2[x]$ of degree at 
most $m$ such that $\tilde{h} \approx_m h^{-1}$.
\end{enumerate}
\end{lemma}
\begin{proof}(Sketch) The first  part is derived from Theorem 7
of \cite{HealyV} by considering the $f_i$ and their product 
as polynomials in $\Fbb_2[x]/(x^{d\ell+1})$. Similarly, for the second
one assume $g \in \Fbb_2[x]/(x^{dm+1}$ and $g^m$ is also in the same ring
without any fear of any distinct powers getting mixed.
The last one (described in the proof of Lemma 20 from \cite{HealyV})
can be proved by considering
$h^{-1} = (1 - (1 -h))^{-1} = 1 + \sum_{i=1}^{\infty}{(1-h)^i}$
Now notice that terms like $(1-h)^i$ for $i > m$ have monomials with exponents
larger than $m$ and hence can be safely discarded. Thus, it suffices to
compute: $1+ \sum_{i=1}^{m}{(1-h)^i}$ in the ring $\Fbb_2[x]/(x^{m+1})$ which
we can from Theorem 7 of \cite{HealyV}.
\end{proof}
\subsection{Static algorithm for Bipartite Maximum Matching from Section~\ref{sec:ac0p}}\label{subsec:isaacProofApp}
The idea behind the proof of \cite{DKKM18} is that:
\begin{enumerate}
\item The most significant weight function $w'$ ensures that the cardinality
of the actual edges (i.e. edges from $G$) picked in the GPM in $G'$ equals the
cardinality of the maximum matching in $G$. This
is because the GPM would cover as many of the
$t_v$ vertices with self-loops as possible to minimize the weight that ensures
the corresponding $v$ must be covered by an actual edge.
\item The next most significant weight function $w''$ is used to ensure
that all the GPMs use the \emph{same} set of pendant edges. This is because, 
otherwise, there is an alternating path in the symmetric difference of the two 
GPMs that starts and ends at self-loops $t_u,t_v$. Then, the difference in the 
weights $w''$ of the two matchings restricted to the path is $u - v \neq 0$ 
and we can find a GPM of strictly smaller weight by replacing the edges of one 
matching with the edges of the other matching restricted to the path, 
contradicting that both matchings were the lightest GPMs.
\item The least significant weight function $w'''$ then isolates the GPM since
all min-weight GPMs are essentially perfect matchings restricted to the same
set $S$ of vertices, namely those that are not matched by the corresponding 
pendant edges and the non-zero circulation weights on $G$ ensures that these 
are isolating weights on the induced graph $G[S]$.
\end{enumerate}

\subsection{Proof from Section~\ref{sec:rank}}\label{subsec:rankApp}
\begin{proof}[Proof of Lemma \ref{lem:combineMatrices}]
Notice that the sets $R,C$ can be sorted in $\AC^0$ because computing the
position $pos_R(r)$ of an element $r \in R$ (i.e., the number of elements not 
larger than $r$) is equivalent to finding the sum of at most $\ell$ bits 
(which are zero for elements of $R$ larger than $r$ and one otherwise).

The position $pos_{\bar{R}}(r')$ of an element $r' \in \bar{R}$ (i.e. the number of
elements of $\bar{R}$ not larger than $r'$) can also be found in $\AC^0$. 
This is because we can first find the set $R(r') = \{r_i \in R: r_i < r'\}$
in $\AC^0$. Then $pos_{\bar{R}}(r') = r' - |R(r')|$ because
there are $r'$ rows with indices at most $r'$ and out of these all but
 $|R(r')|$ are in $\bar{R}$ and thus can be computed in $\AC^0$.
We can similarly compute $pos_{C}(c), pos_{\bar{C}}(c')$ for 
$c \in C$ and $c'\in \bar{C}$.

Finally given $i,j \in [n]$ the element $Y_{i,j}$ is $X^{1,1}_{pos_R(i),pos_C(j)}$
if $i \in R, j \in C$. Similarly if $i \in \bar{R}, j \in C$ then it is
$X^{2,1}_{pos_{\bar{R}}(i),pos_C(j)}$, if $i \in R, j \in \bar{C}$ then it is
$X^{1,2}_{pos_R(i),pos_{\bar{C}}(j)}$ and if $i \in \bar{R}, j \in \bar{C}$
then it is $X^{2,2}_{pos_{\bar{R}}(i),pos_{\bar{C}}(j)}$. This completes the proof.
\end{proof}
\end{document}